\documentclass{bioinfo}
\usepackage{lipsum}

\usepackage[ruled,vlined]{algorithm2e}
\usepackage{algorithmic}

\copyrightyear{2015} \pubyear{2015}

\usepackage{url}
\access{Advance Access Publication Date: Day Month Year}
\appnotes{Manuscript Category}

\begin{document}
\firstpage{1}

\subtitle{Subject Section}

\title[FMtree]{FMtree: A fast locating algorithm of FM-indexes for genomic data}
\author[Cheng \textit{et~al}.]{Haoyu Cheng\,$^{\text{\sfb 1,2,3}}$, Ming Wu\,$^{\text{\sfb 1,2,3}}$ and Yun Xu\,$^{\text{\sfb 1,2,3,}*}$}
\address{$^{\text{\sf 1}}$School of Computer Science, University of Science and Technology of China, Heifei, Anhui, 230027, China\\
$^{\text{\sf 2}}$Key Laboratory on High Performance Computing, Anhui Province\\
$^{\text{\sf 2}}$Collaborative Innovation Center of High Performance Computing, National University of Defense Technology, Changsha, 410073, China}

\corresp{$^\ast$To whom correspondence should be addressed.}

\history{Received on XXXXX; revised on XXXXX; accepted on XXXXX}

\editor{Associate Editor: XXXXXXX}

\abstract{\textbf{Motivation:} 
As a fundamental task in bioinformatics, searching for massive short patterns over a long text has been accelerated by various compressed full-text indexes. These indexes are able to provide similar searching functionalities to classical indexes, e.g., suffix trees and suffix arrays, while requiring less space. For genomic data, a well-known family of compressed full-text index, called FM-indexes, presents unmatched performance in practice. One major drawback of FM-indexes is that their locating operations, which report all occurrence positions of patterns in a given text, are particularly slow, especially for the patterns with many occurrences.\\
\textbf{Results:} 
In this paper, we introduce a novel locating algorithm, FMtree, to fast retrieve all occurrence positions of any pattern via FM-indexes. When searching for a pattern over a given text, FMtree organizes the search space of the locating operation into a conceptual quadtree. As a result, multiple occurrence positions of this pattern can be retrieved simultaneously by traversing the quadtree. Compared with the existing locating algorithms, our tree-based algorithm reduces large numbers of redundant operations and presents better data locality. Experimental results show that FMtree is one order of magnitude faster than the state-of-the-art algorithms, and still memory-efficient.\\
\textbf{Availability:} FMtree is freely available at \url{https://github.com/chhylp123/FMtree}.\\
\textbf{Contact:} \href{xuyun@ustc.edu.cn}{xuyun@ustc.edu.cn}\\
\textbf{Supplementary information:} Supplementary data are available online.}

\maketitle

\section{Introduction}


The string matching problem is to identify the occurrence positions of a short string $P$ (called \emph{pattern}) in a given long string $T$ (called \emph{text}). For genomic data, both the pattern $P$ and the text $T$ are the sequences over a small alphabet $\Sigma = \{a, c, g, t\}$. In order to speed up the string matching process, a well-known approach is \emph{indexed matching}.
It first builds an index data structure for the text in advance, and then searches for the pattern via the index. 
If the text is static, this approach is usually much faster than \emph{online matching} approach, which directly searches for the pattern over the text.
Many bioinformatics applications have adopted the indexed string matching approach, such as read mapping (\citealp{bowtie2}; \citealp{bwa}), genome assembly (\citealp{assembly1}; \citealp{assembly2}) and read error correction (\citealp{correction}).


A number of full-text indexes have been proposed for several decades. 
Classical indexes like suffix arrays (\citealp{suffix_array}),
efficiently support two basic functions: \emph{count} and \emph{locate}. Given a pattern $P$ and a text $T$, the count function is to report the number of occurrences of $P$ in $T$, while the locate function is to retrieve all occurrence positions of $P$ in $T$. A serious problem of these classical indexes is that their space usage is relatively large, especially for a very long text. To address this problem, various compressed full-text indexes have been developed in recent years (\citealp{review2009}).
Generally, most of them are able to be classified into three families: FM-indexes (\citealp{fm-index, fm2}), compressed suffix arrays (CSAs) (\citealp{csa1, csa2}) and Lempel-Ziv compression based indexes (LZ-indexes) (\citealp{lz1, lz2}).
 These indexes are designed to provide similar count and locate functionalities to classical indexes, while requiring less space. For genomic data, the most efficient family of compressed full-text indexes is FM-indexes (\citealp{advantage1, advantage2}), especially when searching for short patterns (\citealp{advantage3}). 
 Given a human genome with about 3.15 billion characters, FM-indexes usually require less than 3GB RAM, while classical suffix arrays require about 12GB RAM. In addition, thanks to the small alphabet size of genomic data, the counting time of FM-indexes is comparable to that of classical indexes (\citealp{drawback, Experience}). Thus, FM-indexes have become the essential data structure in many bioinformatics algorithms (\citealp{bowtie2, bwa, gem}).


The major bottleneck of FM-indexes is that their locating operation is several orders of magnitude slower than that of classical indexes (\citealp{review2009}). To reduce space usage, FM-indexes only save a small fraction of text positions, called \emph{sampled positions}, rather than all of them. When locating a pattern via FM-indexes, the sampled positions of the pattern can be directly retrieved, while the non-sampled positions have to be calculated one-by-one exploiting the expensive LF-mapping operations (see details in Section 2.2). For short patterns with many occurrence positions, FM-indexes need to perform large numbers of LF-mapping operations until all non-sampled positions have been obtained. Unfortunately, searching for short patterns in a long text is an important task in bioinformatics (\citealp{bowtie2, hobbes, BitMapper, mrfast, rHAT}), and these short patterns are very frequent in practice (\citealp{frequence}). In this case, the cost of locating operations dominates the overall string matching time.

Unfortunately, although there are various studies about compressed full-text indexes, only a small fraction of them focus on accelerating the locating operations.
Gonz{\'a}lez {\it et~al}. propose locally compressed suffix array (LCSA) (\citealp{lcsa2, lcsa1}) to improve the data locality of the locating operations. For popular compressed full-text indexes such as FM-indexes and CSAs, a serious problem is that locating patterns via them results in many random memory accesses. To solve this problem, LCSA directly compresses suffix array exploiting the repetitions of suffix array. Since all occurrence positions of a pattern are saved consecutively in suffix array, LCSA can obtain these positions by decompressing consecutive elements of suffix array. In this case, its memory accesses are highly local. 
However, compared with other compressed full-text indexes, LCSA requires much more space when indexing genomic data.

Besides, Ferragina {\it et~al}. develop a distribution-aware algorithm (\citealp{aware}), which adjusts FM-indexes or CSAs according to the distribution of the occurrence positions of query patterns. This algorithm assumes that the distribution of the patterns' occurrence positions has been known in advance, so that it inclines to sample the text positions which have high probability to be located. However, in most cases, it is impossible to know this distribution during the index building phase. Besides, the distribution-aware algorithm cannot achieve good performance unless the distribution of the patterns' occurrence positions is very skewed. This requirement also limits the usage of this algorithm. 

In theory, LZ-indexes are more efficient than FM-indexes and CSAs when performing locating operations (\citealp{review2009}).
But for genomic data, it is difficult to develop a highly optimized LZ-index like the existing sophisticated implementations of FM-indexes in many bioinformatics algorithms (\citealp{bowtie2, bwa, gem}).
Moreover, there does not exist a practical implementation of LZ-indexes which is able to process large texts with billions of characters.

Here we introduce a novel locating algorithm, FMtree, to significantly accelerate the locating operations of FM-indexes for genomic data. When locating a pattern via FM-indexes, the search space of locating operation is organized into a quadtree. By utilizing this quadtree, FMtree is able to calculate the non-sampled positions block-by-block, while current algorithms have to calculate these positions one-by-one. Thus, our tree-based locating algorithm is cache-friendly and avoids many unnecessary operations.
Another advantage of FMtree is that it can be applied to any implementation of FM-indexes without modification. 
Overall experimental results show that FMtree significantly outperforms previous algorithms for genomic data.

\begin{methods}
\section{Background}

\subsection{Definitions and notation}

A string $S$ is a sequence of characters over the alphabet $\Sigma$, and the size of $\Sigma$ is |$\Sigma$|.
We let $|S|$ denote the length of $S$, $S[i]$ denote the $i$-th character of $S$ 
($0 \le i \le |S| - 1$), and $S[i, j]$ denote the substring that starts at $S[i]$ and ends at $S[j]$. 
Besides, consider a character $s$ and a string $S$, $sS$ denotes the concatenation of $s$ and $S$.
We also let $s^n$ denote a string of length $n$ such that $s^0$ is an empty string, $s^1 = s$ and $s^n = ss^{n-1}$.




To solve the string matching problem, existing compressed full-text indexes need to support the following two basic operations:

\begin{itemize}
\item \textbf{count(\emph{P},\emph{T})}: Return the number of occurrences of pattern $P$ in text $T$. 

\item \textbf{locate(\emph{P},\emph{T})}: Return all occurrence positions of pattern $P$ in text $T$. 

\end{itemize}

For the convenience of further discussion, we assume that a special character \$ is at the end of text $T$, 
where \$ is lexicographically smaller than other characters in $\Sigma$.
Since we focus on genomic data, all characters in text $T$ and pattern $P$ belong to $\Sigma = \{a,c,g,t\}$ except $T[|T| - 1] = \$$. 

\subsection{Overview of FM-indexes}

As a family of compressed full-text indexes, FM-indexes are first proposed to emulate classical suffix arrays (\citealp{fm-index}). Given a text $T$, its suffix array $SA$ (\citealp{suffix_array}) saves the positions of all suffixes of $T$ in lexicographic order.
For a pattern $P$, it is obvious that the positions of all suffixes prefixed by $P$ are saved consecutively in an interval of $SA$, called $SA[sp, ep]$. In fact, $SA[sp, ep]$ consists of all occurrence positions of $P$ in $T$. 

\begin{figure}[b]
\centering
\includegraphics[width=0.27\textwidth]{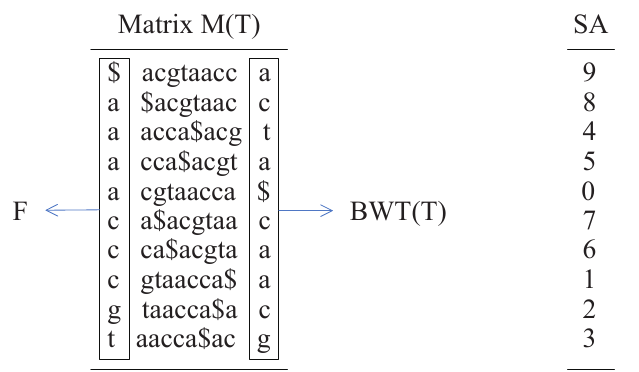}
\caption{An example of Burrows Wheeler Transform $BWT(T)$ and suffix array $SA$ for string $T=$``acgtaacca\$". }
\label{fig1}
\end{figure}

Compared with suffix arrays, 
FM-indexes provide similar searching functionalities, while requiring less space. The critical data structure of a FM-index is Burrows Wheeler Transform (\citealp{BWT}) of $T$, called $BWT(T)$, which permutes the characters of $T$ reversibly. Conceptually, $BWT(T)$ can be constructed in the following two steps:

\begin{itemize}
\item Building a conceptual matrix $M(T)$ including all cyclic rotations of $T$ in lexicographic order. Each row in $M(T)$ is a cyclic rotation of $T$.

\item Let $BWT(T)$ be the last column of $M(T)$.


\end{itemize}

An example is presented in Fig.~\ref{fig1}. 
For FM-indexes, there are two important properties of $M(T)$ and $BWT(T)$:

\begin{itemize}
\item Let $F$ be the first column of $M(T)$, $BWT(T)[i]$ precedes $F[i]$ in text $T$. It is obvious that the $i$-th row of $M(T)$ corresponds to $SA[i]$. Specifically, $BWT(T)[i] = T[SA[i] - 1]$ if $SA[i] \neq 0$, and $BWT(T)[i] = \$$ if $SA[i] = 0$. In addition, $F[i]  = T[SA[i]]$.


\item For any character $s \in \Sigma$, the $i$-th $s$ in $F$ and the $i$-th $s$ in $BWT(T)$ correspond to the same $s$ in $T$. For example in Fig.~\ref{fig1}, 
     $F[5]$ and $BWT(T)[1]$ are the first $c$ in $F$ and $BWT(T)$, respectively. In fact, both of them correspond to $T[7]$.

\begin{algorithm} [b]
        \caption{backward\_search($P$)}
        \begin{algorithmic}[1]
            \REQUIRE the pattern $P[0, |P| - 1]$
            \ENSURE the $SA$ range $[sp, ep]$ of $P$
            \STATE $i = |P| - 1, s = P[i], i=i-1$
            \STATE $sp = C[s], ep = C[s + 1]-1$
            \WHILE{$i \geqslant 0 \mathrel{\&}\mathrel{\&} sp \leqslant ep$}
			\STATE $s = P[i], i=i-1$
			\STATE $sp = C[s] + rank_{s}(BWT(T), sp)$
			\STATE $ep = C[s] + rank_{s}(BWT(T), ep+1)-1$
			\ENDWHILE
			\IF{$sp > ep$}
			\RETURN "not found"
			\ELSE
			\RETURN $[sp, ep]$
			\ENDIF
        \end{algorithmic}
    \end{algorithm}






\end{itemize}

Based on the above properties, a core operation LF (Last-to-First mapping or LF-mapping) of FM-indexes is defined as:


\begin{equation}
LF(l) = C[BWT(T)[l]] + rank_{BWT(T)[l]}(BWT(T), l)
\end{equation}

\noindent where $C[BWT(T)[l]]$ denotes the number of characters in $T$ which are 
smaller than $BWT(T)[l]$, and $rank_{BWT(T)[l]}$$(BWT(T), l)$ is a rank operation that reports the number of $BWT(T)[l]$ in $BWT(T)[0, l- 1]$. 
Essentially, LF operation scans text $T$ backward, namely, $F[l]=T[SA[l]]$, $F[LF[l]]=BWT(T)[l]=T[SA[l]-1]$ and $SA[LF[l]]$ $=SA[l]- 1$. For instance, in Fig.~\ref{fig1}, 
$SA[0]=9$, $LF[0]=1$, $SA[LF[0]] =SA[1]=SA[0]- 1=8$.

\vspace{1\baselineskip}
\noindent\textbf{Counting via FM-indexes.} 
When counting a pattern $P$ in text $T$, it is obvious that all $M(T)$'s rows prefixed by $P$ are saved consecutively in an interval $M(T)[sp, ep]$. Thus, to answer count$(P, T)$, FM-indexes first perform the \emph{backward search} algorithm to determine the range $[sp, ep]$.
After that, count$(P, T) = ep - sp + 1$. Algorithm 1 presents the procedure of the backward search algorithm. More precisely, this algorithm searches $P[0, |P| - 1]$ backward in $|P|$ steps. In $i$-th step ($i = |P| - 1, |P| - 2,\ldots, 0$), this algorithm updates $[sp, ep]$ such that $M(T)[sp, ep]$ includes all $M(T)$'s rows prefixed by $P[i, |P| - 1]$. Note that in first step ($i = |P| - 1$), $[sp, ep]$ is the range of a single character $P[|P| - 1]$, so that it can be directly obtained via array $C$ (line 2 in Algorithm 1).
For any character $s \in \Sigma$, $C[s]$ saves the number of characters which are lexicographically smaller than $s$ in $T$.
Finally, when $i = 0$, the range $[sp, ep]$ of $P[0, |P| - 1]$ is obtained and count$(P, T) = ep - sp + 1$. 
Since $M(T)[i]$ corresponds to $SA[i]$, $SA[sp, ep]$ actually saves all occurrence positions of $P$ in $T$.

\vspace{1\baselineskip}
\noindent\textbf{Locating via FM-indexes.} As mentioned above, the backward search algorithm of FM-indexes determines the range $[sp, ep]$ such that $SA[sp, ep]$ consists of all occurrence positions of $P$ in $T$. If the whole $SA$ is saved, locate$(P, T)$ can be answered by retrieving $SA[sp, ep]$ directly. However, for a long text, its $SA$ is very space-consuming. Thus, FM-indexes only save a fraction of positions in $SA$, called the sampled positions. This strategy reduces the space usage of FM-indexes, but of course comes at the expense of additional computational overhead.
Algorithm 2 presents the most practical locating algorithm of FM-indexes.
To obtain position $SA[i]$ ($sp \le i \le ep$), this algorithm first scans text $T$ backward by performing LF operation $m$ times until a sampled position $SA[j]$ is reached. After that, $SA[i] = SA[j] + m$.

\begin{algorithm} [t]
        \caption{locate($sp, ep$)}
        \begin{algorithmic}[1]
            \REQUIRE the $SA$ range $[sp, ep]$
            \ENSURE the position set $R$ consists of all positions in $SA[sp, ep]$

			\FOR{$i = sp$ to $ep$}
			\STATE $j = i, m = 0$

            \WHILE{$SA[j]$ is not sampled}
			\STATE $j = LF(j)$, $m=m+1$
			\ENDWHILE
			\STATE Add $SA[j]+m$ to $R$
			\ENDFOR
			
			\RETURN $R$
        \end{algorithmic}
    \end{algorithm}








\vspace{1\baselineskip}
\noindent\textbf{Implementations of FM-indexes.} As we can see, both the counting operation and the locating operation of FM-indexes can be reduced to rank operations. To support rank operations, practical implementations of FM-indexes break $BWT(T)$ into small blocks. For the beginning line of each block and each character $s \in \Sigma$, the rank value is precomputed and saved.
When calculating $rank_{s}(BWT(T),l)$, these implementations first retrieve the precomputed rank value of the block which includes $BWT(T)[l]$, 
and then add the number of rest $s$ in this block before $l$-th line. Besides, the locating algorithm of FM-indexes requires a sampled suffix array $SSA$, which saves all sampled positions in suffix array order.
Several implementations of FM-indexes also need a bitmap $B$, where $B[i] = 1$ denotes that $SA[i]$ is saved in $SSA$ ($0 \le i \le |T| - 1$). If $B[i] = 1$, $SA[i]$ can be found in $SSA[rank_{1}(B, i)]$.
We refer to the surveys (\citealp{review2009, review2007}) for more information about the implementations of FM-indexes.

\section{Methods}

\subsection{Analysis of Existing Locating Algorithm}

The first performance bottleneck is that existing locating algorithm needs to perform a large number of LF operations. 
As mentioned in Section 2.2, FM-indexes sample the positions in $SA$ to reduce the space usage. A popular sampling strategy is to sample every $SA[i]$ if $SA[i]$ mod $ D = 0$, where $D$ is the regular sampling distance. Here we refer to this strategy as \emph{value sampling} strategy.
This strategy guarantees that any occurrence position of a pattern can be obtained in at most $D - 1$ steps of LF operation.
Thus, to locate a pattern with $occ$ occurrence positions, the number of LF operations is $(D - 1) \times occ$ in worst case. Unfortunately, short and frequent patterns with large value of $occ$ are widely used in practice, which results in massive LF operations.
Apart from value sampling strategy, many FM-index-based bioinformatics algorithms 
adopt another sampling strategy, called \emph{subscript sampling} strategy. This strategy samples every $SA[i]$ if $i$ mod $D = 0$, where $D$ is the regular sampling distance. Compared with value sampling strategy, subscript sampling strategy leads to even worse locating performance. The reason is that for any position, value sampling strategy can guarantee to obtain it in at most $D - 1$ steps of LF operation, while subscript sampling strategy cannot.

Besides, the poor data locality is another bottleneck of existing locating algorithm. To locate a position $SA[i_{0}]$, existing locating algorithm needs to perform $m$ steps of LF operation until a sampled position $SA[i_{m}]$ is reached.
Essentially, this procedure scans text $T$ backward from $SA[i_{0}]$ to $SA[i_{m}]$, so that $SA[i_{m}] = SA[i_{0}] - m$.
In $t$-th step ($t = 1, 2,\ldots, m$), the aim of existing locating algorithm is to calculate $i_{t}$ such that $SA[i_{t}] = SA[i_{0}] - t = SA[i_{t - 1}] - 1$.
As shown in Algorithm 2, $i_{t}$ is obtained by calculating $LF(i_{t - 1})$, which needs to access $BWT(T)[i_{t - 1}]$.
In addition, if a FM-index is sampled by value sampling strategy, in $t$-th step, we need to access $B[i_{t-1}]$ to check if $SA[i_{t-1}]$ is a sampled position.
Thus, the memory access addresses to $BWT(T)$ and $B$ in $t$-th step are determined by $i_{t - 1}$. Similarly, in $(t + 1)$-th step, the memory access addresses are determined by $i_{t}$.
If the suffix $T[SA[i_{t - 1}], |T|-1]$ is not lexicographically similar to the suffix $T[SA[i_{t}], |T|-1]$, $i_{t - 1}$ and $i_{t}$ would be very different.
In this case, the memory accesses of these two steps are non-contiguous.
For example, consider a non-sampled position $SA[1003]$ of pattern $P = $ ``acgagt'' and $BWT(T)[1003] = c$.
To obtain $SA[1003]$, existing locating algorithm checks $B[1003]$ and calculates $LF(1003)$ in first step. Since $BWT(T)[1003] = T[SA[1003] - 1] = c$ and $SA[LF(1003)] = SA[1003] - 1$, $SA[LF(1003)]$ is actually an occurrence position of $cP =$ ``\underline{c}acgagt''. Obviously, ``\underline{c}acgagt'' is significantly lexicographically larger than ``acgagt'', so that $LF(1003)$ is very different to 1003. In second step, existing locating algorithm of FM-indexes needs to check $B[LF(1003)]$ and calculate $LF(LF(1003))$. Therefore, the memory access addresses in first step and second step to $BWT(T)$ and $B$ are non-contiguous.

\subsection{Our Proposed Algorithm: FMtree}

The key idea of FMtree is to organize the search space of the locating operation into a conceptual quadtree, so that multiple locations can be located simultaneously by traversing the quadtree. This idea is based on the observation that, different occurrence positions of a pattern $P$ may be obtained by performing similar LF operations.
For example, consider two non-sampled positions $SA[2136]$ and $SA[2137]$ of pattern $P =$ ``acgagt'', and $BWT(T)[2136]=BWT(T)[2137]=t$. In first step, existing locating algorithm calculates $LF(2136)$ for $SA[2136]$, and calculates $LF(2137)$ for $SA[2137]$. Since $BWT(T)[2136] = BWT(T)[2137] = t$, both $SA[LF(2136)]$ and $SA[LF(2137)]$ are the occurrence positions of $tP =$ ``\underline{t}acgagt''. As a result, $LF(2136)$ and $LF(2137)$ are very similar. In second step, for $SA[2136]$ and $SA[2137]$, existing locating algorithm calculates $LF(LF(2136))$ and $LF(LF(2137))$, respectively. Thus, if we locate $SA[2136]$ and $SA[2137]$ jointly, namely, we calculate $LF(2136)$ and $LF(2137)$ together in first step, and calculate $LF(LF(2136))$ and $LF(LF(2137))$ together in second step, the memory accesses in each step to $BWT(T)$ and $B$ would be contiguous.

More precisely, given two positions $SA[i]$ and $SA[j]$ with $BWT(T)[i] = BWT(T)[j] = s$, if there does not exist $s$ in $BWT(T)[i + 1, j - 1]$, it is obvious that $LF(i) = LF(j) - 1$.
Thus, for the positions in $SA[sp, ep]$ with same character $s$ in $BWT(T)$, the results of one step of LF operation belong to the range $[LF(sp_s), LF(ep_s)]$, where $SA[sp_s]$ and $SA[ep_s]$ are the first position and last position in $SA[sp, ep]$ with character $s$ in $BWT(T)$. If $SA[sp, ep]$ consists of all occurrence positions of pattern $P$, $SA[LF(sp_s), LF(ep_s)]$ actually includes all occurrence positions of $sP$.
Therefore, $[LF(sp_s), LF(ep_s)]$ = backward\_search$(sP) = [C[s] + rank_{s}(BWT(T), sp), C[s] + rank_{s}(BWT(T), ep+1)-1]$, as shown in Algorithm 1. 
In other words, to obtain all positions in $SA[sp, ep]$, large numbers of LF operations for all positions can be reduced to a few rank operations only for $sp$ and $ep$. Formally, this is based on the following theorem:

\newtheorem{theorem}{Theorem}
\begin{theorem}
Given a text $T$ over alphabet $\Sigma$ and its FM-index which is sampled by value sampling strategy with sampling distance $D$. Let $FM(P, T)$ be the sampled position set including all sampled occurrence positions of $P$ in FM-index, and $L(P, T)$ be the position set including all occurrence positions of $P$ in $T$. Then $L(P, T)$ can be calculated as follows:


\begin{equation}
L(P, T)=\bigcup_{i = 0}^{D - 1} \{ x| x=y+i,y \in FM(*^iP, T)\}
\end{equation}

\noindent where * is a wildcard of $\Sigma$.

\end{theorem}

\begin{proof}
Generally, all positions in $T$ are classified into $D$ sets: $g_{0}, g_{1}, \ldots, g_{(D-1)}$. Each set $g_{i}$ consists of every position $SA[j]$ of $T$, where $SA[j]$ mod $D = i$. Let $Lg_{i}(P, T)$ be the position set including all occurrence positions of $P$ in $g_{i}$. Note that $FM(P, T) = Lg_{0}(P, T)$. Obviously, $Lg_{i}(P, T)$ can be obtained as:

\begin{equation}
Lg_{i}(P, T)=\{ x| x=y+i,y \in FM(*^{i}P, T)\}
\end{equation}
Since $L(P, T)$ consists of all occurrence positions of $P$ in $D$ sets $(g_{0},$ $g_{1}, \ldots, g_{(D-1)})$, $L(P, T)$ is:
\begin{equation}
L(P, T)=\bigcup_{i = 0}^{D - 1}Lg_{i}(P, T)
\end{equation}$\verb+                                               +$
\end{proof}

\vspace{1\baselineskip}
\noindent\textbf{Basic algorithm of FMtree.} According to Theorem 1, we propose the core algorithm of FMtree. To utilize FMtree, FM-indexes must be sampled by value sampling strategy with regular sampling distance $D$.
When locating a pattern $P$ via FM-indexes, there are total $D$ steps in FMtree.
Specifically, in $i$-th step $(i = 0, 1,\ldots, D - 1)$, this algorithm consists of the following three stages:

\begin{itemize}
\item FMtree first searches $*^iP$ via FM-indexes to obtain their corresponding $SA$ ranges. In total, there are |$\Sigma$|$^i$ $SA$ ranges in $i$-th step, since $*^iP$ represents |$\Sigma$|$^i$ different strings.
For each string $S[0, |S| - 1]$ in $i$-th step $(i = 1,..., D - 1)$, its $SA$ range $[sp, ep]$ is updated from the $SA$ range $[sp_{1}, ep_{1}]$ of string $S[1, |S| - 1]$ in $(i - 1)$-th step by two rank operations to $BWT(T)$.
More precisely, $sp = C[S[0]] + rank_{S[0]}(BWT(T), sp_{1})$, and $ep = C[S[0]] + rank_{S[0]}(BWT(T), ep_{1}+1)-1$, as shown in Algorithm 1.
Note that the $SA$ range of $P$ in 0-th step has been calculated in advance exploiting the backward search algorithm. 


\item Then FMtree needs to retrieve the sampled positions in these $SA$ ranges.
As shown in Section 2.2, FM-index saves all sampled positions of $SA$ in $SSA$ in suffix array order, and utilizes bitmap $B$ such that $B[i] = 1$ denotes that $SA[i]$ is saved in $SSA$. Therefore, given a $SA$ range $[sp, ep]$, all sampled positions in $SA[sp, ep]$ are saved consecutively in $SSA[ssp, sep]$, where $ssp = rank_{1}(B, sp)$ and $sep = rank_{1}(B, ep)$.

\item Finally, once all sampled positions of $*^iP$ have been obtained, FMtree adds $i$ to these positions to obtain the occurrence positions of $P$.

\end{itemize}

With the above three stages, in $i$-th step, FMtree is able to obtain all positions in $Lg_{i}(P, T)$.
Thus, $L(P, T)$ can be obtained in total $D$ steps of FMtree.
For genomic data with alphabet size $|\Sigma|$ = 4, the search space of FMtree is actually a quadtree of height $D$, as shown in Fig.~\ref{fig2}. 
Indeed, the $i$-th step of FMtree corresponds to the $i$-th layer of this quadtree. We observe that FMtree locates all occurrence positions of $P$ block-by-block, while existing locating algorithm has to locate these positions one-by-one.

\begin{figure}[t]
\centering
\includegraphics[width=0.5\textwidth]{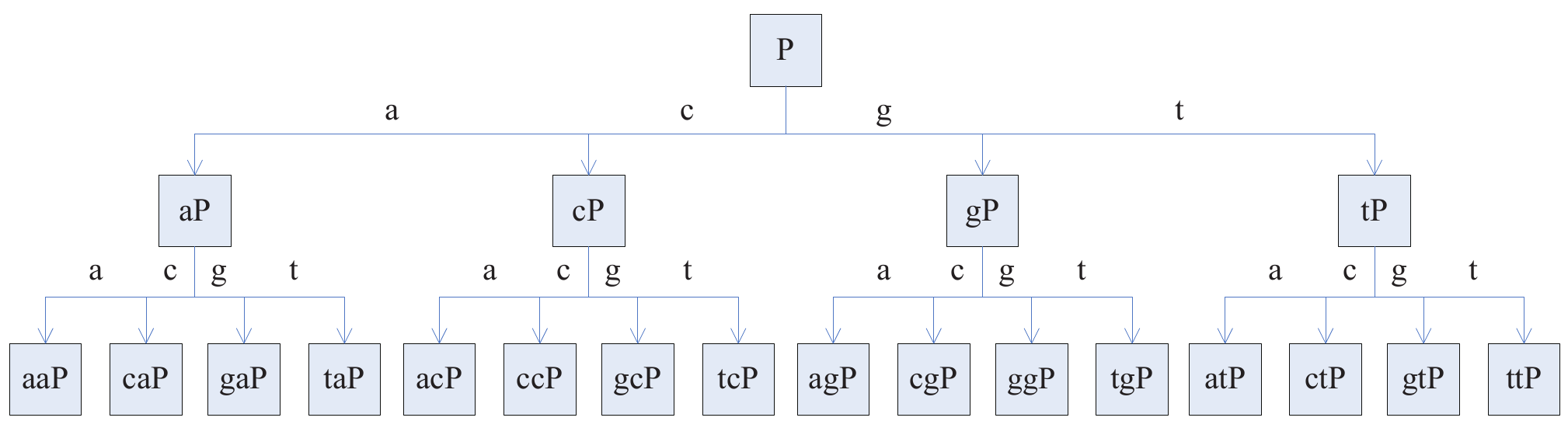}
\caption{An example illustrates the search space of the basic FMtree when the sampling distance
$D = 3$ and the pattern is $P$. }
\label{fig2}
\end{figure}

\vspace{1\baselineskip}
\noindent\textbf{Further optimizations of FMtree.}  The key problem of FMtree is that with the increasing value of sampling distance $D$, the number of rank operations increases exponentially.
In $i$-th step, the number of rank operations to $BWT(T)$ is $2 \times |\Sigma|^i = 2 \times 4^i$.
In addition, an equal number of rank operations to $B$ are also required to determine the $SSA$ ranges.
As such there are many rank operations when $D$ is large.

To solve this problem, we first propose an optimization to FMtree, called \emph{early leaf nodes calculation}.
For the basic FMtree, we observe that the cost of calculating its leaf nodes (i.e., the $(D-1)$-th step of FMtree) dominates its overall locating time. The reason is that in $(D-1)$-th step, the number of rank operations to both $BWT(T)$ and $B$ is 2 $\times 4^{D-1}$, which is larger than the total number of rank operations in the rest $D-1$ steps of FMtree. Early leaf nodes calculation is proposed to avoid the expensive $(D-1)$-th step of FMtree. It is based on the following theorem:

        \begin{theorem}
    Given a text $T[0, |T|-1]$ and a pattern $P[0, |P|-1]$, the position set $L(P, T)$, which includes all occurrence positions of $P$ in $T$, can be calculated as follows:


    \begin{equation}
    \begin{split}
    L(P, T)=\bigcup_{i = 1}^{D} \{ x| x=y-i, y \in FM(P[i, |P|-1], T), y \ge i, \\ P[0, i-1] = T[y-i,y-1]\}
    \end{split}
    \end{equation}

    \end{theorem}

    \begin{proof}
    Since $FM(P[i, |P|-1], T)=Lg_{0}(P[i, |P|-1], T)$, for each position $SA[j]$ in $FM(P[i, |P|-1], T)$, $(SA[j]-i)$ mod $D = D-i$. Then $Lg_{(D-i)}(P[0, |P|-1], T)$ can be obtained as:
\begin{equation}
    \begin{split}
    Lg_{(D-i)}(P, T)= \{ x| x=y-i, y \in FM(P[i, |P|-1], T),\\ y \ge i,  P[0, i-1] = T[y-i,y-1]\}
    \end{split}
    \end{equation}

    Therefore, $L(P, T)$ is:
    \begin{equation}
    L(P, T)=\bigcup_{i = 1}^{D}Lg_{(D-i)}(P, T)
    \end{equation}$\verb+                                               +$
    \end{proof}

Actually, the aim of the $(D-1)$-th step in basic FMtree is to obtain $Lg_{(D-1)}(P, T)$. According to Theorem 2, this position set can be obtained using early leaf nodes calculation in following three stages:

\begin{itemize}
\item Early leaf nodes calculation first searches $P[1, |P|-1]$ via FM-indexes to obtain its corresponding $SA$ range $[sp_1, ep_1]$.


\item Then for every $SA[j]$ ($sp_1 \le j \le ep_1$) if $B[j]=1$ and $T[SA[j]-1]=P[0]$, early leaf nodes calculation adds it to a position set $R$. Note that since $T[SA[j]-1]= BWT(T)[j]$, early leaf nodes calculation actually checks all elements in $B[sp_1, ep_1]$ and $BWT(T)[sp_1, ep_1]$. Thus, the memory accesses in this stage are highly local. In contrast, the $(D-1)$-th step of basic FMtree results in many random memory accesses. A more practical and efficient implementation about this stage can be found in Supplementary Section S1.

\item Finally, early leaf nodes calculation subtracts 1 from the positions in $R$ which has been obtained in second stage.

\end{itemize}

Apart from early leaf nodes calculation, two simple branch-cut strategies are proposed to further improve the performance of FMtree (from 0-th step to $(D-2)$-th step). First, for an interval $SA[sp_i, ep_i]$ in $i$-th step ($0 \le i \le D-2$), if $ep_i - sp_i + 1$ is smaller than a predefined threshold, FMtree calculates all positions in $SA[sp_i, ep_i]$ one-by-one in at most $D-i-2$ steps of LF operation.
Thus, the number of LF operations for $SA[sp_i, ep_i]$ in following $D - i - 2$ steps is $(ep_i - sp_i + 1) \times (D-i-2)$ in worst case. Note that the rank operation to $BWT(T)$ is the dominant cost of LF operation. When $D$ is large, if FMtree does not adopt this branch-cut strategy, the number of rank operations to $BWT(T)$ in following $D - i - 2$ steps is $2 \times (4^1+4^2+\ldots+4^{D-i-2}) = 8 \times (4^{D-i-2}-1)/3$, which is much larger than $(ep_i - sp_i + 1) \times (D-i-2)$. Similarly, the number of rank operations to $B$ can also be reduced.
Second, when locating a pattern with $occ$ occurrence positions, FMtree terminates once $occ$ occurrence positions have been obtained.


\begin{algorithm} [t]
        \caption{FMtree($P, sp, ep, sp_1, ep_1, D$)}
        \begin{algorithmic}[1]
            \REQUIRE the pattern $P[0, |P|-1]$; the $SA$ range $[sp, ep]$ of $P[0, |P|-1]$;
            the $SA$ range $[sp_1, ep_1]$ of $P[1, |P|-1]$; the sampling distance $D$
            \ENSURE the position set $R$ consists of all positions in $SA[sp, ep]$
            \STATE $total\_num=ep-sp+1$, $num=0$
            \STATE perform early leaf nodes calculation exploiting $[sp_1, ep_1]$ and $P[0]$; add the obtained $m$ positions to $R$; $num=num+m$
            \STATE $tree\_height=D-1$
            \STATE $node.sp=sp, node.ep=ep, node.layer=0$
            \STATE $Queue.EnQueue(node)$

            \WHILE{$Queue$ is not empty $\mathrel{\&}\mathrel{\&} num<total\_num$}
            \STATE $Queue.DeQueue(node)$
            \STATE $sp=node.sp, ep=node.ep, layer=node.layer$
            \IF{$ep-sp+1<threshold$}
            \STATE calculate positions in $SA[sp, ep]$ one-by-one in at most $tree\_height-layer-1$ steps of LF operation; add the obtained $m$ positions to $R$; $num=num+m$.
            \ELSE
            \STATE $ssp = rank_{1}(B, sp), sep = rank_{1}(B, ep)$
            \STATE $num=num+sep-ssp+1$
            \FOR{$k = ssp$ to $sep$}
			\STATE Add $SSA[k]+layer$ to $R$
            \ENDFOR

            \IF{$layer+1<tree\_height$}
            \STATE$spchild[0] = C[a] + rank_{a}(BWT(T), sp)$
            \STATE$spchild[1] = C[c] + rank_{c}(BWT(T), sp)$
            \STATE$spchild[2] = C[g] + rank_{g}(BWT(T), sp)$
            \STATE$spchild[3] = C[t] + rank_{t}(BWT(T), sp)$
            \STATE$epchild[0] = C[a] + rank_{a}(BWT(T), ep+1)-1$
            \STATE$epchild[1] = C[c] + rank_{c}(BWT(T), ep+1)-1$
            \STATE$epchild[2] = C[g] + rank_{g}(BWT(T), ep+1)-1$
            \STATE$epchild[3] = C[t] + rank_{t}(BWT(T), ep+1)-1$

            \FOR{$t = 0$ to $3$}
			\STATE $node.sp=spchild[t], node.ep=epchild[t]$
			\STATE $node.layer=layer+1$
			
			\STATE $Queue.EnQueue(node)$
            \ENDFOR
            \ENDIF

            \ENDIF
			
			\ENDWHILE

			\RETURN $R$
        \end{algorithmic}
    \end{algorithm}















\vspace{1\baselineskip}
\noindent\textbf{Full algorithm of FMtree.} Algorithm 3 presents the full algorithm of FMtree. By utilizing a queue data structure \emph{Queue}, the conceptual quadtree of FMtree is traversed in breadth-first order. In fact, the height of this quadtree is $D-1$ instead of $D$ (line 3 in Algorithm 3). This is because early leaf nodes calculation (line 2 in Algorithm 3) is used to avoid the $(D-1)$-th step in basic FMtree. For any node in quadtree, the $SA$ ranges of its four children are calculated jointly to improve the data locality (line 18-25 in Algorithm 3).
A detailed analysis is presented in Supplementary Section S2.

\end{methods}

\section{Results}

In our experiments, we used the following three datasets:

\begin{itemize}
\item Dna.200MB consists of 209.72 million characters from Pizza\&Chili corpus (\url{http://pizzachili.dcc.uchile.cl/}), which is the standard benchmark in compressed full-text indexes (\citealp{review2009}). 

\item For practical bioinformatics algorithms, their indexes must be able to process large texts with billions of characters. Thus, the human genome including 3.16 billion characters was used in our experiments.

\item Another large text is the mouse genome. It consists of 2.73 billion characters.
\end{itemize}

These datasets were used as texts in our experiments. Like popular bioinformatics algorithms (\citealp{bwa, soap2}), character $n$ in these three datasets was converted to one of $a$, $c$, $g$ and $t$ randomly.

We first compare FMtree with two state-of-the-art methods, including locally compressed suffix array (LCSA) (\citealp{lcsa1, lcsa2}) and LZ-index (\citealp{lz1}). LCSA is designed specifically to accelerate the locating operation of compressed full-text indexes, and LZ-index has been proven that it is very competitive in locating speed (\citealp{review2009}).
Besides, we implemented two FM-index-based locating algorithms: Original\_v and Original\_s. The only difference between Original\_v and Original\_s is their sampling strategies. Original\_v locates patterns via the FM-index sampled by value sampling strategy, while Original\_s locates patterns via the FM-index sampled by subscript sampling strategy (see Section 3.1). For FMtree, Original v and Original s, we implemented a highly optimized FM-index for genomic data. Its main data structures ($C$ and $BWT(T)$) are similar to those in (\citealp{soap2}), which is a well-established FM-index-based bioinformatics algorithm (see Supplementary Section S3). 
Note that FMtree, Original\_v and Original\_s are independent on any particular implementation of FM-indexes, so that we did not test these locating algorithms with other implementations of FM-indexes. Another family of compressed full-text indexes CSAs was not tested in our experiments. The reason is that for genomic data, previous studies (\citealp{review2009, Experience}) have shown that CSAs cannot outperform FM-indexes, LCSA and LZ-indexes for locating operation. For detailed description about the experimental setting, please see Supplementary Section S3.

\subsection{Comparison on Small Text}

In the first experiment, dna.200MB was used as text to evaluate the performance of different methods.
As previous studies about the locating operations (\citealp{review2009, Experience}), patterns were generated by randomly selecting 10 short substrings of length 5 from the text. By utilizing these short patterns with many occurrence positions, we could focus on the performance of locating operations and ignore the influence of counting operations. The reason is that compared with the locating time of different methods, their counting time was negligible in this experiment.
For LZ-index, a parameter $\epsilon$ trades the locating time for space usage.
With the increasing value of $\epsilon$, the space usage of LZ-index decreases, but its locating time increases. We set $\epsilon = \{1, 2 , 3, 4\}$ in this experiment.
For LCSA, its default parameters were directly used. 
For FM-indexes used in FMtree, Original\_v and Original\_s, we set the sampling distance $D = \{2, 3, 4, 5, 6, 7, 8\}$, to make the space usage of FM-indexes similar to that of LZ-index and LCSA.

As shown in Fig.~\ref{fig3}, with similar space usage, FMtree is one or two order of magnitude faster than other methods. Apart from FMtree, LCSA outperforms other methods in locating speed. However, it requires much more space than the others. 
We also observe that the space usage of Original\_s is slightly less than that of FMtree and Original\_v. Unlike Original\_s which uses subscript sampling strategy, FMtree and Original\_v adopt value sampling strategy. Thus, FMtree and Original\_v need an extra bitmap $B$ to mark all sampled positions in $SA$. In exchange, Original\_s is about 1.5 times slower than Original\_v, and at least 40 times slower than FMtree.

\begin{figure}
\centering
\includegraphics[width=0.40\textwidth]{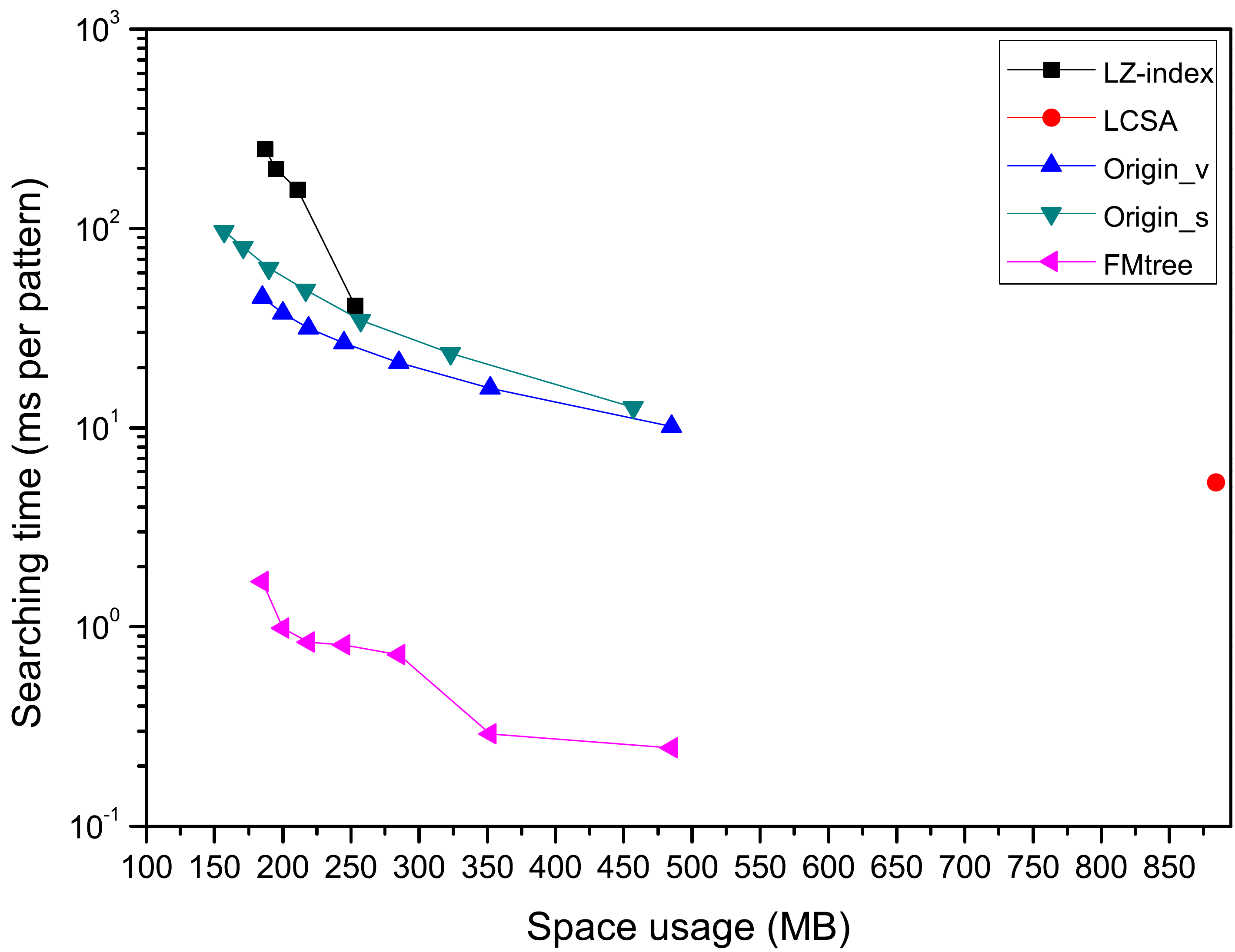}
\caption{Time/space tradeoffs of different methods for dna.200MB.}
\label{fig3}
\end{figure}

\begin{figure}[b]
\centering
\includegraphics[width=0.50\textwidth]{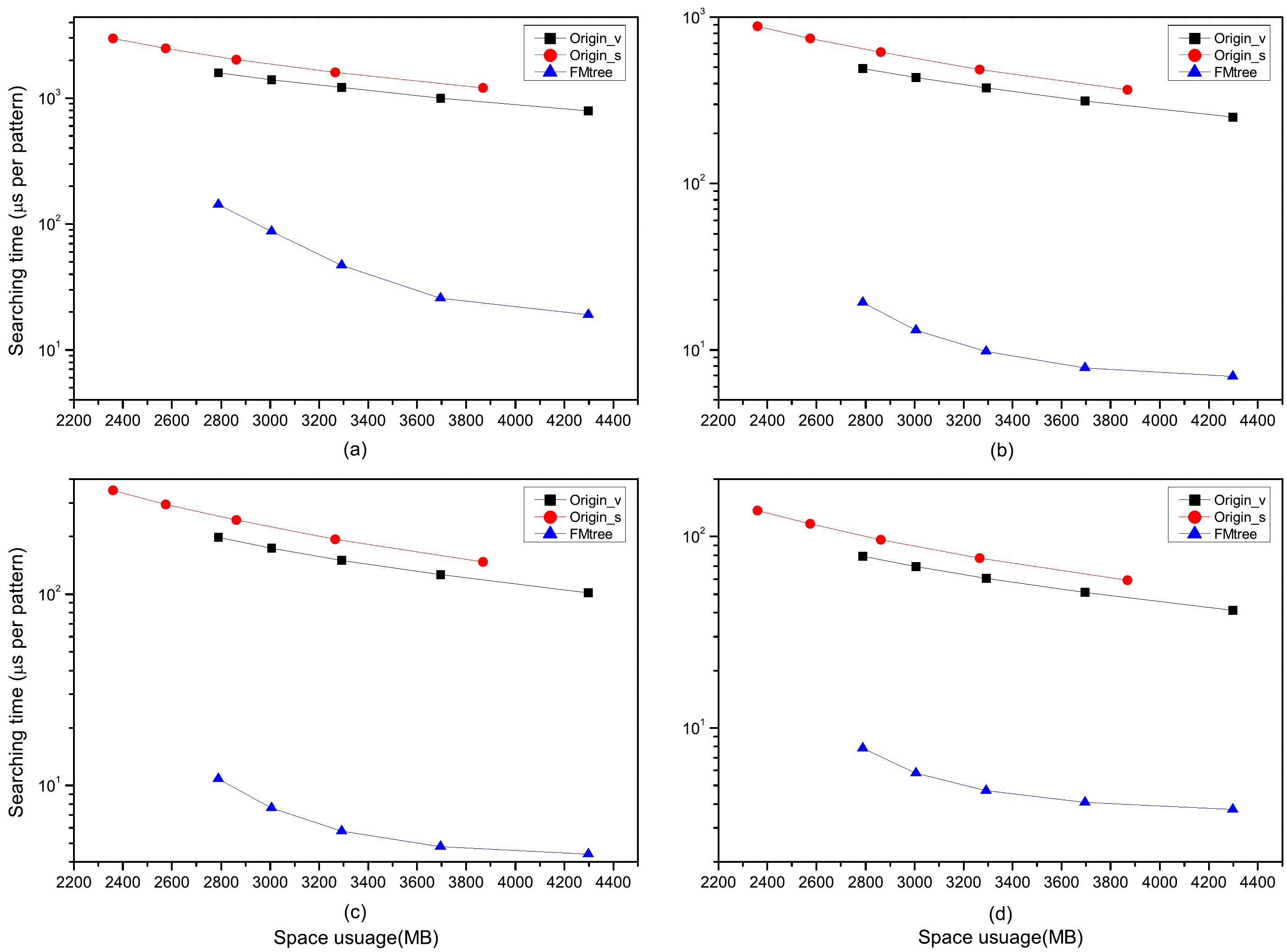}
\caption{Time/space tradeoffs of different methods for human genome. For subgraph (a), (b), (c) and (d), the length of pattern is 12, 16, 20 and 25, respectively.}
\label{fig4}
\end{figure}

\subsection{Comparison on Large and Practical Texts}

In the second experiment, we studied the performance of different methods on two large texts: human genome and mouse genome.
To generate patterns, we randomly extracted short substrings from both human genome and mouse genome.
For each text, we generated four datasets including 100k patterns of length 12, 16, 20 and 25, respectively.
In fact, short patterns of length 10 to 25 are widely used in existing bioinformatics algorithms (\citealp{mrfast, bowtie2}). Note that in this experiment, we did not tested LZ-index and LCSA due to two reasons. First, their implementations cannot process large texts like human genome and mouse genome. Second, the results in first experiment have shown that they are significantly slower than FMtree.

Fig.~\ref{fig4} and Fig.~\ref{fig5} present the results of FMtree, Original\_v and Original\_s with sampling distance $D = \{4, 5, 6, 7, 8\}$.  
With the increasing value of $D$, all methods require less space.
Note that in this experiment, the space usage of FM-index is close to that of the input human genome and mouse genome, which require 3GB RAM and 2.6GB RAM, respectively.
For human genome (see Fig.~\ref{fig4}), FMtree is up to 62 times faster than Original\_v and Original\_s. And for mouse genome (see Fig.~\ref{fig5}), FMtree is at most 86 times faster than other methods.
We also present the locating time of different methods in Table S1 and Table S2. With respect to Original\_v and Original\_s, FMtree achieves highest gain when we focus on the locating time instead of the overall searching time. This is because in addition to the locating operation, the searching algorithm of FM-index also consists of the counting operation.

\begin{figure}[t]
\centering
\includegraphics[width=0.50\textwidth]{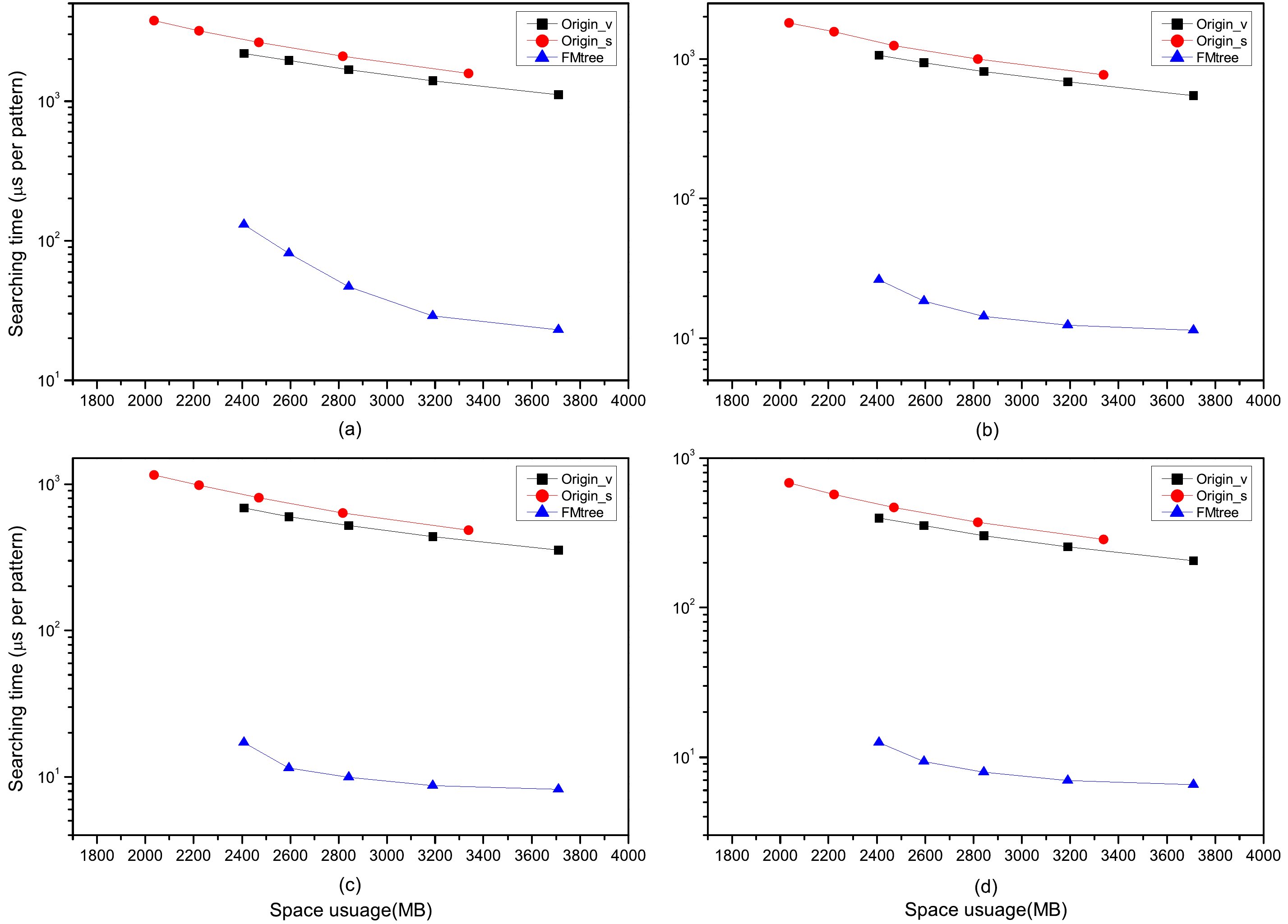}
\caption{Time/space tradeoffs of different methods for mouse genome. For subgraph (a), (b), (c) and (d), the length of pattern is 12, 16, 20 and 25, respectively.}
\label{fig5}
\end{figure}

\section{Conclusion and future work}

In this paper we propose a novel algorithm, FMtree, to accelerate the locating operations of FM-indexes for genomic data. When searching for a pattern via FM-indexes, FMtree builds a conceptual quadtree, so that multiple occurrence positions of the pattern can be obtained simultaneously by traversing this quadtree. In contrast, existing locating algorithm has to calculate all occurrence positions one-by-one. Therefore, FMtree reduces massive unnecessary operations and presents better data locality. We also introduce several strategies to further speed up FMtree.

For genomic data with small alphabet size, FMtree is significantly faster than state-of-the-art methods in our experiments. However, the performance of FMtree decreases rapidly with the increasing size of alphabet. In the future, it would be interesting to optimize FMtree for the applications with large alphabet size.

\section*{Acknowledgements}

This work was partially supported by the National Nature Science Foundation of China under the grant No. 61672480
and the Program for Excellent Graduate Students in Collaborative Innovation Center of High Performance
Computing.

%
%

\begin{thebibliography}{}














\bibitem[Ahmadi {\it et~al}., 2012]{hobbes}
Ahmadi, A. {\it et~al}. (2012) Hobbes: optimized gram-based methods for efficient read alignment. {\it Nucleic acids research}, {\bf 40(6)}, e41-e41.

\bibitem[Arroyuelo, 2006]{lz1}
Arroyuelo, D. {\it et~al}. (2006) Reducing the space requirement of LZ-index. {\it Annual Symposium on Combinatorial Pattern Matching}, 318-329.

\bibitem[Arroyuelo {\it et~al}., 2012]{lz2}
Arroyuelo, D. {\it et~al}. (2012) Stronger Lempel-Ziv based compressed text indexing. {\it Algorithmica}, {\bf 62}, 54-101.



\bibitem[Burrows and Wheeler, 1994]{BWT}
Burrows, M. and Wheeler, D. J. (1994) A block-sorting lossless data compression algorithm. {\it Technical Report 124}, Digital Equipment Corporation, California.


\bibitem[Cheng {\it et~al}., 2015]{BitMapper}
Cheng, H. {\it et~al}. (2015) BitMapper: an efficient all-mapper based on bit-vector computing. {\it BMC bioinformatics}, {\bf 16(1)}, 192.



\bibitem[Deorowicz and Grabowski, 2013]{drawback}
Deorowicz, S. and Grabowski, S. (2013) Data compression for sequencing data. {\it Algorithms for Molecular Biology}, {\bf 8(1)}, 25.


\bibitem[Ferragina and Manzini, 2000]{fm-index}
Ferragina, P. and Manzini, G. (2000) Opportunistic data structures with applications. In {\it Foundations of Computer Science, 2000. Proceedings. 41st Annual Symposium}, 390-398.


\bibitem[Ferragina {\it et~al}., 2009]{review2009}
Ferragina, P. {\it et~al}. (2009) Compressed text indexes: From theory to practice. {\it Journal of Experimental Algorithmics (JEA)}, {\bf 13}, 12.




\bibitem[Ferragina {\it et~al}., 2013]{aware}
Ferragina, P. {\it et~al}. (2013) Distribution-aware compressed full-text indexes. {\it Algorithmica}, {\bf 67(4)}, 529-546.



















\bibitem[Gog and Petri, 2014]{Experience}
Gog, S. and Petri, M. (2014) Optimized succinct data structures for massive data. {\it Software: Practice and Experience}, {\bf 44(11)}, 1287-1314.

\bibitem[Gog {\it et~al}., 2015]{extended}
Gog, S. {\it et~al}. (2015) Improved and extended locating functionality on compressed suffix arrays. {\it Journal of Discrete Algorithms}, {\bf 32}, 53-63.

\bibitem[Gog {\it et~al}., 2017]{advantage1}
Gog, S. {\it et~al}. (2017) CSA++: Fast Pattern Search for Large Alphabets. {\it 2017 Proceedings of the Ninteenth Workshop on Algorithm Engineering and Experiments (ALENEX)}, 73-82.

\bibitem[Gonz{\'a}lez and Navarro, 2007]{lcsa2}
Gonz{\'a}lez, R. and Navarro, G. (2007) Compressed text indexes with fast locate. {\it Annual Symposium on Combinatorial Pattern Matching}, 216-227.

\bibitem[Gonz{\'a}lez {\it et~al}., 2015]{lcsa1}
Gonz{\'a}lez, R. {\it et~al}. (2015) Locally compressed suffix arrays. {\it Journal of Experimental Algorithmics (JEA)}, {\bf 19}, 1-1.




\bibitem[Grabowski {\it et~al}., 2004]{fm2}
Grabowski, S. {\it et~al}. (2004) First Huffman, then Burrows-Wheeler: A simple alphabet-independent FM-index. {\it International Symposium on String Processing and Information Retrieval}, 210-211.

\bibitem[Grossi and Vitter, 2005]{csa1}
Grossi, R. and Vitter, J. S. (2005) Compressed suffix arrays and suffix trees with applications to text indexing and string matching. {\it SIAM Journal on Computing}, {\bf 35}, 378-407.


\bibitem[Hach {\it et~al}., 2010]{mrfast}
Hach, F. {\it et~al}. (2010) mrsFAST: a cache-oblivious algorithm for short-read mapping. {\it Nature methods}, {\bf 7(8)}, 576-577.


\bibitem[Hon {\it et~al}., 2004]{advantage3}
Hon, W.K. {\it et~al}. (2004) Practical aspects of Compressed Suffix Arrays and FM-Index in Searching DNA Sequences. {\it ALENEX/ANALC}, 31-38.





\bibitem[Langmead and Salzberg, 2012]{bowtie2}
Langmead, B. and Salzberg, S. L. (2012) Fast gapped-read alignment with Bowtie 2. {\it Nature methods}, {\bf 9(4)}, 357-359.





\bibitem[Li, 2012]{assembly1}
Li, H. (2012) Exploring single-sample SNP and INDEL calling with whole-genome de novo assembly. {\it Bioinformatics}, {\bf 28(14)}, 1838-1844.



\bibitem[Li, 2013]{bwa}
Li, H. (2013) Aligning sequence reads, clone sequences and assembly contigs with BWA-MEM. {\it arXiv preprint}, {\bf arXiv:1303.3997}.


\bibitem[Li {\it et~al}., 2009]{soap2}
Li, R. {\it et~al}. (2009) SOAP2: an improved ultrafast tool for short read alignment. {\it Bioinformatics}, {\bf 25(15)}, 1966-1967.


\bibitem[Liu {\it et~al}., 2015]{rHAT}
Liu, B. {\it et~al}. (2015) rHAT: fast alignment of noisy long reads with regional hashing. {\it Bioinformatics}, {\bf 32(11)}, 1625-1631.



\bibitem[Manber and Myers, 1993]{suffix_array}
Manber, U. and Myers, G (1993) Suffix arrays: a new method for on-line string searches. {\it SIAM Journal on Computing}, {\bf 22(5)}, 935-948.

\bibitem[Marco-Sola {\it et~al}., 2012]{gem}
Marco-Sola, S. {\it et~al}. (2012) The GEM mapper: fast, accurate and versatile alignment by filtration. {\it Nature methods}, {\bf 9}, 1185-1188.


\bibitem[Navarro and M{\"a}kinen, 2007]{review2007}
Navarro, G and M{\"a}kinen, V. (2007) Compressed full-text indexes. {\it ACM Computing Surveys (CSUR)}, {\bf 39}, 2.

\bibitem[Sadakane, 2003]{csa2}
Sadakane, K. {\it et~al}. (2003) New text indexing functionalities of the compressed suffix arrays. {\it Journal of Algorithms}, {\bf 48}, 294-313.


\bibitem[Schulz {\it et~al}., 2014]{correction}
Schulz, M. H. {\it et~al}. (2014) Fiona: a parallel and automatic strategy for read error correction. {\it Bioinformatics}, {\bf 30(17)}, i356-i363.


\bibitem[Simpson and Durbin, 2012]{assembly2}
Simpson, J. T. and Durbin, R. (2012) Efficient de novo assembly of large genomes using compressed data structures. {\it Genome research}, {\bf 22(3)}, 549-556.









\bibitem[Vyverman {\it et~al}., 2012]{advantage2}
Vyverman, M. {\it et~al}. (2012) Prospects and limitations of full-text index structures in genome analysis. {\it Nucleic acids research}, {\bf 40}, 6993-7015.






\bibitem[Xin {\it et~al}., 2016]{frequence}
Xin, H. {\it et~al}. (2016) Optimal seed solver: optimizing seed selection in read mapping. {\it Bioinformatics}, {\bf 32(11)}, 1632-1642.


































































































\end{thebibliography}

\end{document}